\documentclass[twocolumn,pra,superscriptaddress]{revtex4-2}
\usepackage{amsmath,amssymb,amsthm,easybmat,mathtools}
\usepackage{xcolor}
\usepackage{graphicx}
\usepackage{dcolumn}
\usepackage{bm}
\usepackage{hyperref}
\hypersetup{colorlinks=true, linkcolor=blue!80!black, urlcolor=blue!80!black, citecolor=blue!80!black}

\usepackage[nameinlink]{cleveref}

\usepackage{comment,subcaption} 
\usepackage{tikz}
\usetikzlibrary{positioning}

\newtheorem{proposition}{Proposition}

\newcommand{\mbb}{\mathbb}
\newcommand{\mc}{\mathcal}

\newcommand{\tr}{\textrm{Tr}}

\newcommand{\ket}[1]{|#1\rangle}
\newcommand{\bra}[1]{\langle #1|}
\newcommand{\op}[2]{|#1\rangle\langle #2|}
\newcommand{\ip}[2]{\langle #1| #2 \rangle}

\newcommand{\wt}{\widetilde}

\definecolor{cool_green}{rgb}{0.0, 0.5, 0.0}

\newcommand{\todo}[1]{{\color{red} #1}}

\begin{document}

\preprint{APS/123-QED}

\title{Exploring Non-Multiplicativity in the Geometric Measure of Entanglement}


\author{Daniel Dilley}
\email{ddilley@anl.gov}
\affiliation{Mathematics and Computer Science Division, Argonne National Laboratory}
 \author{Jerry Chang}
 \email{chiehc3@illinois.edu}
 \affiliation{Department of Electrical and Computer Engineering and Coordinated Science Laboratory, University of Illinois at Urbana-Champaign}
\author{Jeffrey Larson}
 \email{jmlarson@anl.gov}
\affiliation{Mathematics and Computer Science Division, Argonne National Laboratory}
\author{Eric Chitambar}
\email{echitamb@illinois.edu}
\affiliation{Department of Electrical and Computer Engineering and Coordinated Science Laboratory, University of Illinois at Urbana-Champaign}
\affiliation{Illinois Quantum Information Science and Technology (IQUIST) Center, University of Illinois Urbana-Champaign}

\date{\today}

\begin{abstract}
The geometric measure of entanglement (GME) quantifies how close a multi-partite quantum state is to the set of separable states under the Hilbert-Schmidt inner product.  
The GME can be non-multiplicative, meaning that the closest product state to two states is entangled across subsystems.  
In this work, we explore the GME in two families of states: those that are invariant under bilateral orthogonal $(O\otimes O)$ transformations, and mixtures of singlet states.  In both cases, a region of GME non-multiplicativity is identified around the anti-symmetric projector state.  We employ state-of-the-art numerical optimization methods and models to quantitatively analyze non-multiplicativity in these states for $d=3$.  We also investigate a constrained form of GME that measures closeness to the set of \textit{real} product states and show that this measure can be non-multiplicative even for real separable states.

\end{abstract}

\maketitle


\section{Introduction}

One of the first entanglement measures studied in the literature is the geometric measure of entanglement (GME).  First proposed by Shimony for the case of bipartite pure states \cite{Shimony-1995a} and later extended to multipartite states by Barnum and Linden \cite{Barnum-2001a}, the GME measures how closely a given state resembles a product state under the Hilbert-Schmidt inner product.  That is, for an $N$-partite state $\ket{\Psi}^{A_1\cdots A_N}$, its GME is defined as
    \begin{equation}
     \Lambda^2(\Psi^{A_1\cdots A_N})=\max_{\ket{a_1,a_2,\cdots,a_N}}|\ip{\Psi}{a_1,a_2,\cdots,a_N}|^2,
 \end{equation}
 where the maximization is taken over all product states $\ket{a_1,a_2,\cdots,a_N}\equiv \ket{a_1}\otimes\ket{a_2}\otimes\cdots\otimes\ket{a_N}$. Subsequent works have explored the analytical properties of GME, explicitly calculating its value for special types of states and connecting it to other entanglement measures \cite{Wei-2003a, Wei-2004a, Cavalcanti-2006a, Markham-2007a, Hayashi-2008a, Tamaryan-2009a, Hubener-2009a, Jung-2008a, Zhu-2011a}.
 
For the case of tripartite states, the GME can be computed by analyzing the reduced density matrix of any two parties. 
This involves extending the GME measures to mixed states through the definition $\Lambda^2(\rho^{A_1\cdots A_N})=\max_{\ket{a_1,\cdots,a_N}}\bra{a_1,\cdots,a_N}\rho\ket{a_1,\cdots,a_N}$.  One thus obtains
\begin{align}
    \Lambda^2(\Psi^{A_1\cdots A_N})=\Lambda^2(\rho^{A_1\cdots A_{N-1}}),
\end{align}
where $\rho^{A_1\cdots A_{N-1}}=\tr_{A_N}\Psi^{A_1\cdots A_N}$.  While $\Lambda^2$ is a \textit{bona fide} entanglement measures for pure states, this is no longer the case when evaluated on mixed states. 
This is because $\Lambda^2$ is a monotone under local operations and classical communication (LOCC) only when restricted to pure states; for the general case, one must use a convex-roof extended version of $\Lambda^2$ to recover LOCC monotonicity \cite{Vidal-2000a}. 

In this paper we consider the notoriously difficult question of GME multiplicativity.  For two states $\rho$ and $\sigma$, their GME is said to be multiplicative if $\Lambda^2(\rho\otimes\sigma)=\Lambda^2(\rho)\Lambda^2(\sigma)$.  For any pair of bipartite pure states, multiplicativity of the GME follows directly from their Schmidt decomposition \cite{Shimony-1995a}.  However, outside of special cases like this, the GME is no longer guaranteed to be multiplicative.  The first known example of non-multiplicativity was demonstrated by \citet{Werner-2002a}.  Hayden and Winter \cite{Hayden-2008a} proved the existence of states that are strongly non-multiplicative (i.e. $\Lambda^2(\rho\otimes\sigma)>\!>\Lambda^2(\rho)\Lambda^2(\sigma)$).   \citet{Zhu-2011a} later established non-multiplicativity for different anti-symmetric multipartite states, and further, they furnished a random sampling argument implying that certain families of multipartite states generically have non-multiplicative GME.  However, it still remains an open challenge to decide whether the GME is multiplicative for a given pair of states, even in small dimensions with few parties.

An equivalent form of our question can be phrased in terms of the maximal output entropy of a completely-positive (CP) map.  Letting $\mc{D}(A)$ and $\mc{D}(B)$ denote the set of density matrices on systems $A$ and $B$, respectively, and $\Vert\rho\Vert_p=[\tr(\rho^p)]^{1/p}$ denoting the Schatten $p$-norm of $\rho\in\mc{D}(B)$, the maximal output $p$-purity of a CP map $\mc{N}:\mc{D}(A)\to\mc{D}(B)$ is defined as $\gamma_p(\mc{N})=\max_{\ket{a}\in A}\Vert \mc{N}(\op{a}{a})\Vert_p$.  Taking the limit $p\to \infty$ yields the spectral norm $\Vert\rho\Vert_\infty=\max_{\ket{b}}\bra{b}\rho\ket{b}$.  In this case, the maximal output $\infty$-purity is given by
\begin{align}
    \gamma_\infty(\mc{N})=\max_{\ket{a}\in A, \ket{b}\in B}\bra{b}\mc{N}(\op{a}{a})\ket{b}.
\end{align}
To connect this to the GME, we use the Choi-Jamio{\l}kowski operator $J_{\mc{N}}:=\text{id}\otimes\mc{N}(\phi^+)$, where $\phi^+=\sum_{i,j=1}^{|A|}\op{ii}{jj}$, and note the relationship $\mc{N}(\op{a}{a})=\tr_A[(\op{a^*}{a^*}^A\otimes\mbb{I}^B) J_{\mc{N}}]$ with $X^*$ denoting the complex conjugate of the Hermitian operator $X$.  Then by inspection we obtain 
\begin{align}
    \gamma_\infty(\mc{N})=\Lambda^2(J_{\mc{N}}).
\end{align}
Conversely, any bipartite density matrix $\omega^{AB}$ can be uniquely identified with the CP map $\mc{N}_\omega$ through the relationship $\mc{N}_\omega(\rho)=\tr_A[(\rho^*\otimes\mbb{I})\omega^{AB}]$ \cite{Watrous-2018a}.  From this, we can phrase the question of GME multiplicativity in terms of maximal output purity multiplicativity:
\begin{align}
\Lambda^2(\omega\otimes\sigma)&=\Lambda^2(\omega)\Lambda^2(\sigma) \notag\\\;\Leftrightarrow\;\;\gamma_\infty(\mc{N}_\omega\otimes\mc{N}_\sigma)&=\gamma_\infty(\mc{N}_\omega)\gamma_\infty(\mc{N}_\sigma).
\end{align}
Questions on channel additivity and multiplicativity have been heavily studied throughout the development of quantum information theory.  Turning particularly to the works of King, it has been shown that $\gamma_\infty$ is multiplicative whenever one of the channels acts on a qubit system \cite{King-2006a} or it is entanglement-breaking \cite{King-2002a}.  In terms of the Choi-Jamio{\l}kowski operator, qubit channels correspond to two-qubit states and entanglement-breaking channels correspond to separable states.  Note that the proofs for multiplicativity of $\gamma_\infty$ in \cite{King-2006a} and \cite{King-2002a} consider only channels that are trace-preserving maps. On the level of bipartite density matrices, this trace-preserving assumption translates into a condition that (by symmetry) one of the subsystems is maximally mixed.  In \Cref{app:more}, we show that the trace-preserving assumption can be relaxed for entanglement-breaking channels and the multiplicativity of $\gamma_\infty$ still holds.  More generally, we give a direct proof that the GME is strongly multiplicative for any multipartite separable state. To provide a non-exhaustive summary on some of the known sufficient conditions for GME multiplicativity in bipartite systems, a state $\rho^{AB}$ satisfies $\Lambda^2(\rho^{AB}\otimes \sigma^{AB})=\Lambda^2(\rho^{AB})\Lambda^2(\sigma^{AB})$ for any other state $\sigma^{AB}$ if one or more of the following holds:
\begin{enumerate}
    \item[(a)] $\rho^{AB}$ is a two-qubit state with one system having a maximally mixed reduced density matrix \cite{King-2002a};
    \item[(b)] $\rho^{AB}$ is a separable state \cite{King-2006a} (and \Cref{app:more});
    \item[(c)] $\rho^{AB}$ can be written in some basis as a matrix with non-negative elements \cite{Zhu-2011a}.
    \item[(d)] $\rho^{AB}$ is a mixture of certain maximally entangled states \cite{Datta-2005a}.
\end{enumerate}
\noindent To our knowledge, it is still unknown whether multiplicativity of GME holds for all two-qubit states $\rho^{AB}$, although we have not been able to find a counterexample through extensive numerical searches.

In this work, we study the multiplicativity of GME in two-qutrit systems.  The first known example of non-multiplicativity was found by Holevo and Werner \cite{Werner-2002a} in the normalized projector onto the $3\otimes 3$ anti-symmetric subspace.  This is one extreme point in the family of $U\otimes U$ invariant states, the so-called Werner states, given in arbitrary $d\otimes d$ systems by
\begin{align}
    \gamma^{AB}_\lambda=\frac{2(1-\lambda)}{d(d+1)}\pi^+ + \frac{2\lambda}{d(d-1)}\pi^-,
\end{align}
where $\pi^{\pm}=(\mbb{I}\otimes\mbb{I}\pm\mbb{F})/2$ are the symmetric and
anti-symmetric projectors, and $\mbb{F}=\sum_{i=1}^d\op{ji}{ij}$ is the SWAP
operator. Zhu \textit{et al.} provide an intuitive explanation for the fact that 
$\Lambda^2(\gamma_1\otimes\gamma_1)\geq\Lambda^2(\gamma_1)\Lambda^2(\gamma_1)$
for all $d\geq 3$ \cite{Zhu-2011a}.  Roughly speaking, anti-symmetric states are more entangled than symmetric ones, and since two copies of an anti-symmetric state is globally symmetric, one might expect the latter to sit relatively closer to the set of unentangled states.  
Here we further explore this intuition by studying the GME in two families of $3\otimes 3$ states.  The first is the class of $O\otimes O$ invariant states with $O$ being any element from the orthogonal group $O(3)$.  The second family is the convex hull of all two-level anti-symmetric states, that is, states of the form $\frac{1}{\sqrt{2}}(\ket{ij}-\ket{ji})$ for any $i\neq j$.  In both of these families, we discover new instances of GME non-multiplicativity around the anti-symmetric projector $\pi^-$.  To find these examples, we use a powerful numerical optimization technique described in \Cref{Sect:3x3-examples}.  
In \Cref{Sect:real-GME} we consider a restricted form of the GME where maximization is carried out with respect to \textit{real} product vectors: vectors whose coefficients are real-valued in the computational basis. Curiously, we find that the multiplicativity of GME no longer holds for separable states, even for two-qubit systems and for density matrices that are themselves real-valued.  Such results are of particular interest in real theories of quantum mechanics (e.g., \cite{Wootters_2014}) and resource theories of imaginarity (e.g., \cite{Hickey-2018a, Wu-2021a}).

\section{Non-Multiplicativity of GME in two Qutrits}

\label{Sect:3x3-examples}

To find instances of non-multiplicativity, we combine analytical and numerical approaches.  We first analytically compute the GME of a given two-qutrit state $\rho^{AB}$, which is feasible due to the high symmetry in the states we consider. We then numerically explore the two-copy multiplicativity of these states. This is done by simplifying the global objective function 
\begin{equation}
  \bra{\psi,\phi}\rho^{AB}\otimes\rho^{A'B'}\ket{\psi,\phi}
  \label{eq:objective}
\end{equation}
for parametrized product vectors $\ket{\psi}^{AA}\otimes\ket{\phi}^{BB'}$, and then using an optimizer to search for the largest value. All Mathematica files to run the optimizations are uploaded to GitHub \cite{github_dilley}. 
 
In addition to Mathematica's built-in heuristic optimization approaches, we also model our numerical optimization problems in the General Algebraic Modeling System (GAMS)~\cite{gams} and use the BARON~\cite{baron} rigorous global optimization solver, which uses a branch-and-reduce algorithm to maintain both best-found GME (lower bounds) and best-possible GME values (upper bounds) throughout the maximization process. By improving both of these values, BARON provides a certificate of optimality that is not provided by most other optimization routines.

When the optimizer identifies a state where \eqref{eq:objective} exceeds the largest-possible value of $\Lambda^2(\rho^{AB})$, then we have definitively detected an instance of non-multiplicativity.  On the other hand, if the optimizer terminates with a proof that outputs an upper bound of \eqref{eq:objective} equal to $\Lambda^2(\rho^{AB})$, then the state has multiplicative GME. For difficult problem instances, there can be gap between the upper and lower bounds on \eqref{eq:objective} found by the optimizer due to a finite running time. Our experiments use relative threshold of $10^{-5}$ to conclude when the GME is multiplicative when the optimizer returns an upper bound within this threshold. Technically, it may be possible that the GME is non-multiplicative in these instances, but the violation must be very small.

\subsection{($O\otimes O$)-invariant states}

\label{Sect:OO-states}

A bipartite $d\otimes d$ state $\rho^{AB}$ is called orthogonally invariant if $(O\otimes O) \rho^{AB}(O\otimes O)^\dagger=\rho^{AB}$ for every $d$-dimensional orthogonal matrix $O$; i.e. $OO^T=\mbb{I}$ where $O^T$ is the matrix transpose of $O$ in some fixed basis.  The collection of all such invariant states forms a two-parameter family given by
\begin{align}
    \omega_{x,y}&=\frac{2x}{(d-1)(d+2)}(\pi^+-\Phi^+)+\frac{2y}{d(d-1)}\pi^-\notag\\ \label{Eq:OO_Symmetric_States}
    &+(1-x-y)\Phi^+,
\end{align}
where $\Phi^+=\frac{1}{d}\sum_{i,j=1}^d\op{ii}{jj}$ \cite{Vollbrecht-2001a}.  Positivity of $\omega_{x,y}$ requires that $x,y\geq 0$ and $x+y\leq 1$.  To compute the GME of $\omega_{x,y}$, we use the relationships 
\begin{align}
    \bra{\alpha,\beta}\mbb{F}\ket{\alpha,\beta}&=|\ip{\alpha}{\beta}|^2,\notag\\
    \bra{\alpha,\beta}\Phi^+\ket{\alpha,\beta}&= \dfrac{1}{d}|\ip{\alpha^*}{\beta}|^2
\end{align}
for arbitrary vectors $\ket{\alpha}^A$ and $\ket{\beta}^B$.  Then a straightforward calculation shows that
\begin{align}
    &\bra{\alpha,\beta}\omega_{x,y}\ket{\alpha,\beta}\notag\\
    =&\;\frac{2x}{(d-1)(d+2)}\left[\frac{1+|\ip{\alpha}{\beta}|^2}{2}-\frac{|\ip{\alpha^*}{\beta}|^2}{d}\right] \notag\\
    &+\frac{2y}{d(d-1)}\left[\frac{1-|\ip{\alpha}{\beta}|^2}{2}\right]+(1-x-y)\frac{|\ip{\alpha^*}{\beta}|^2}{d}\notag\\
    =&\;\frac{2y +d(x+y)}{d(d-1)(d+2)}+\frac{-2y +d(x-y)}{d(d-1)(d+2)}|\ip{\alpha}{\beta}|^2 \notag\\
    &+\left[\frac{1-y}{d}-\frac{(d+1)x}{(d-1)(d+2)}\right]|\ip{\alpha^*}{\beta}|^2.
    \label{Eq:OO-product-optimization}
\end{align}

Note that $0\leq |\ip{\alpha}{\beta}|,|\ip{\alpha^*}{\beta}|\leq 1$, and furthermore, we can simultaneously attain any combination of these extreme points.  For example, the choice $\ket{\alpha}=\ket{\beta}=\frac{1}{\sqrt{2}}(\ket{1}+i\ket{2})$ attains $\ip{\alpha}{\beta}=1$ and $\ip{\alpha^*}{\beta}=0$, while the choice $\ket{\alpha}=\ket{\beta}=\ket{1}$ attains $\ip{\alpha}{\beta}=\ip{\alpha^*}{\beta}=1$.  The maximum of \eqref{Eq:OO-product-optimization} over normalized product vectors $\ket{\alpha,\beta}$ will then be obtained by choosing $\ip{\alpha}{\beta},\ip{\alpha^*}{\beta}\in\{0,1\}$ based on the sign of its respective coefficient.  We therefore conclude that
\begin{align}
    \label{Eq:GME-OO}
\Lambda^2(\omega_{x,y})=\max\begin{cases}\frac{2y +d(x+y)}{d(d-1)(d+2)},\\ \frac{2x}{(d-1)(d+2)},\\
    \frac{1}{d}-\frac{x d}{(d-1)(d+2)}-\frac{y(d-2)}{d(d-1)}, \\\frac{1-y}{d}-\frac{x}{(d+2)}\end{cases}.
\end{align}

As a natural candidate for the two-copy GME, we consider two copies of $\ket{\Phi^+}=\frac{1}{\sqrt{d}}\sum_{i=1}^{d}\ket{ii}$, held on systems $AA'$ and $BB'$ respectively.  In the following calculation, we obtain a lower bound on the geometric measure:
\begin{align}
    &\bra{\Phi^+}^{AA'}\bra{\Phi^+}^{BB'}\omega_{xy}^{AB}\otimes\omega_{xy}^{A'B'}\ket{\Phi^+}^{AA'}\ket{\Phi^+}^{BB'}\notag\\ \label{Eq:Global_phiP}
    =&\;\dfrac{2x^2 }{d^2 (d-1) (d+2)} + \dfrac{2y^2}{d^3 (d-1)} +\dfrac{(1-x-y)^2}{d^2}.
\end{align}
To make a comparison, we focus on the special one-parameter subfamily defined by the choice $x=0$.  In this case, we can use \eqref{Eq:GME-OO} and \eqref{Eq:Global_phiP} to conclude
\begin{align} 
    \Lambda^2(\omega_{0,y})^2&=\left(\dfrac{1}{d} - \dfrac{y(d-2)}{d(d-1)}\right)^2 \label{Eq:Local}  \\
    \Lambda^2(\omega_{0,y}^{\otimes 2})&\geq\dfrac{2y^2}{d^3 (d-1)} + \dfrac{(1-y)^2}{d^2} \label{Eq:Global}
\end{align}
The plot in \Cref{Fig:Crossover_Plot} shows the $y$-value where \eqref{Eq:Local} equals to the lower bound of \eqref{Eq:Global} for dimensions 3 through 15. Note that the crossover value of $y$ approaches 1 as the Hilbert space dimension $d$ increases.

\begin{figure}[h!]
    \centering
    \includegraphics[width=0.86\linewidth]{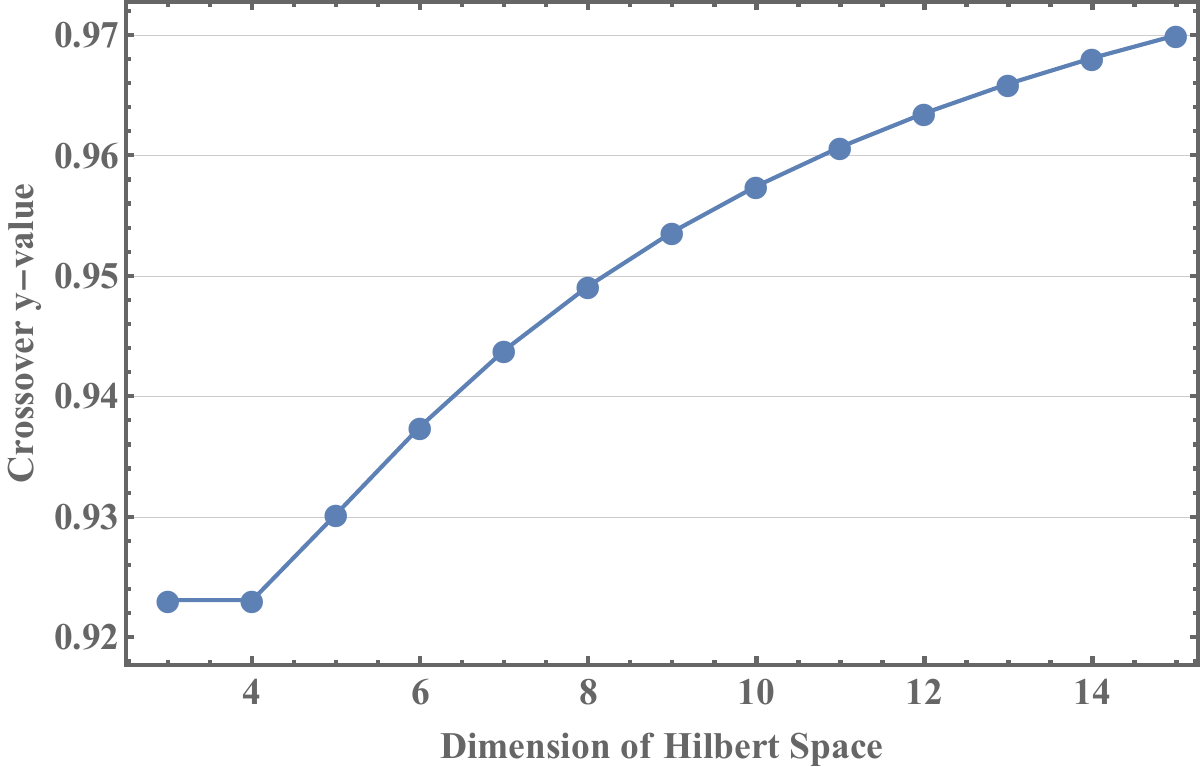}
    \caption{Crossover values of $y$ as a function of $d$ for when the fidelity of $\omega_{0,y}^{\otimes 2}$ with $\ket{\Phi^+}^{AA'}\ket{\Phi^+}^{BB'}$ exceeds $\Lambda^{2}(\omega_{0,y})^2$.  Points above the line correspond to positive instances of non-multiplicativity.}
    \label{Fig:Crossover_Plot}
\end{figure}

To explore the full family $\omega_{x,y}$, we use our numerical optimization.  The results for $\Lambda^2(\omega_{x,y}^{AB}\otimes\omega_{x,y}^{A'B})$ in the case of $d=3$ are presented in \Cref{Fig:OO_Complex_Optimization}.  The separability region for $\omega_{x,y}$ is known to be $1\geq x+y\geq\frac{2}{3}$ and $y\leq\frac{1}{2}$ \cite{Vollbrecht-2001a}, and so multiplicativity of GME holds for these states.  But we also observe multiplicativity for the majority of entangled states as well.  On the other hand, we were able to find explicit instances of non-multiplicativity at 29 out of the 861 points in the approximate region defined by the trapezoid with boundaries $x = 0, x = 0.225, y = 0.95 - x, $ and $y = 1-x$.

\begin{figure}[h!]
           \centering
        \includegraphics[width=.35\textwidth]{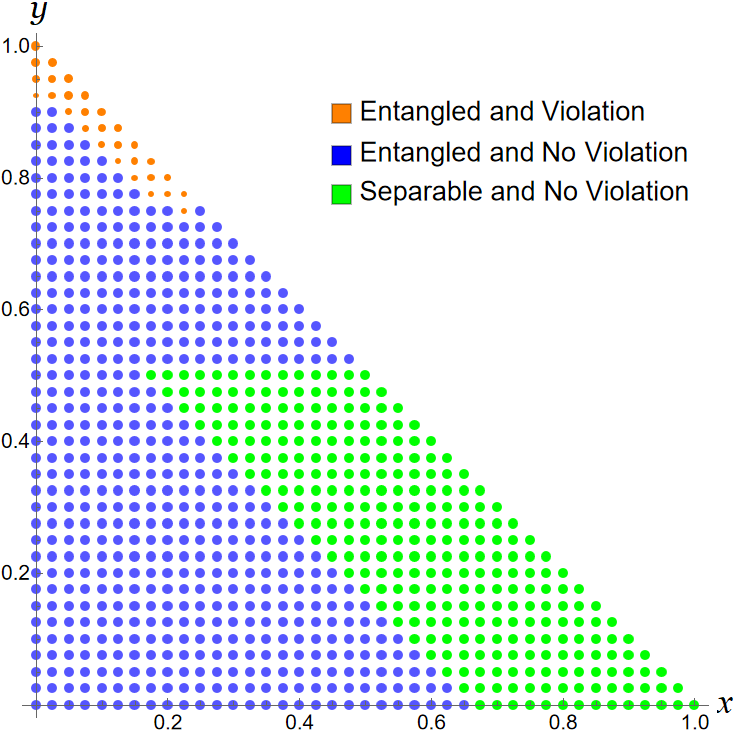}
        \caption{Numerical values of $\Lambda^2(\omega_{x,y}^{AB}\otimes\omega_{x,y}^{A'B'})$ for $d=3$.}
        \label{Fig:OO_Complex_Optimization}
\end{figure}

\subsection{Two-level anti-symmetric states}

The next class of states we consider is built from two-level anti-symmetric states.  For any $i\neq j$, let $\ket{\Psi^{-}_{ij}}=\frac{1}{\sqrt{2}}(\ket{ij}-\ket{ji})$ denote the singlet state in the subspace $\text{span}\{\ket{i}\otimes\ket{i},\ket{i}\otimes\ket{j},\ket{j}\otimes\ket{i},\ket{j}\otimes\ket{j}\}$.  Any $d\otimes d$ probabilistic mixture of such states has the form
\begin{align}
\label{Eq:Singlet_State_Mixture}
    \tau^{AB}=\sum_{j>i = 1}^d p_{ij}\op{\Psi^-_{ij}}{\Psi^-_{ij}}.
\end{align}
To compute the GME of $\tau^{AB}$, let $\ket{\alpha}=\sum_{i=1}^d a_i\ket{i}$ and $\ket{\beta}=\sum_{j=1}^d b_j\ket{j}$ be arbitrary states.  Then
\begin{align}
    \bra{\alpha,\beta}\tau^{AB}\ket{\alpha,\beta}=\frac{1}{2}\sum_{j>i = 1}^dp_{ij}|a_ib_j-a_jb_i|^2.
\end{align}
Extrema of this function can be found by first noticing that
\begin{align}
    \sum_{j>i = 1}^dp_{ij}|a_ib_j-a_jb_i|^2 \leq p^{\#} \sum_{j>i = 1}^d |a_ib_j-a_jb_i|^2
\end{align}
where $p_{ij} \leq p^\#$ for all $i$ and $j$. Now we expand the term in the summand as
\begin{align} \label{Eq:Expanded_Form}
    |a_i|^2 |b_j|^2 + |a_j|^2 |b_i|^2 - 2 Re(a_i a_j^* b_j b_i^*)
\end{align}
and reduce these terms even further. For instance, we see that the first two terms can be reduced to
\begin{align} \notag
    \sum_{j>i}^d& (|a_i|^2 |b_j|^2 + |a_j|^2 |b_i|^2) = \sum_{i \neq j}^d |a_i|^2 |b_j|^2 \\ \notag
    = &\left( \sum_i^d |a_i|^2 \right) \left( \sum_j^d |b_j|^2 \right)    - \sum_i^d |a_i|^2 |b_i|^2 \\ \label{Eq:Reduction_1}
    = & 1 - \sum_i^d |a_i|^2 |b_i|^2
\end{align}
due to normalization of the vectors $\ket{\alpha}$ and $\ket{\beta}$. The last term in (\ref{Eq:Expanded_Form}) can be written as
\begin{align}
    - 2 \sum_{j>i}^d Re(a_i a_j^* b_j b_i^*) = -\sum_{i \neq j}^d Re(a_i a_j^* b_j b_i^*),
\end{align}
but we know that
\begin{align}
    \sum_{i \neq j}^d a_i a_j^* b_j b_i^* 
    = |\ip{\alpha}{\beta}|^2 - \sum_i^d |a_i|^2 |b_i|^2
\end{align}
so that
\begin{align} \label{Eq:Reduction_2}
    - 2 \sum_{j>i}^d Re(a_i a_j^* b_j b_i^*) = -|\ip{\alpha}{\beta}|^2 + \sum_i^d |a_i|^2 |b_i|^2.
\end{align}
Summing the terms in (\ref{Eq:Reduction_1}) and (\ref{Eq:Reduction_2}), we are able to show that (\ref{Eq:Expanded_Form}) is equivalent to the more desirable form of $1 - |\ip{\alpha}{\beta}|^2$, which is bounded by 1; hence, the inequality
\begin{align}
    \sum_{j>i = 1}^dp_{ij}|a_ib_j-a_jb_i|^2 \leq p^{\#}.
\end{align}
Moreover, we can always reach this local optimum by taking inner products with unit vectors that are orthogonal; that is, $\ket{\alpha} = \ket{i}$ and $\ket{\beta} = \ket{j}$ if $p^{\#} = p_{ij}$. Therefore, we conclude that the GME is given by 
\begin{equation}
    \Lambda^2(\tau^{AB}) = \frac{1}{2}\text{max}\{p_{ij}\}.
\end{equation}
Restricting to the case of $d=3$, and using the definitions $p_{12} = x, p_{13} = y,$ and $p_{23} = 1 - x - y$, gives
\begin{align}
\label{Eq:GME-convex-asym}
\Lambda^2(\tau_{x,y}^{AB})=\frac{1}{2}\max\{x,y,1-x-y\}.
\end{align}
Note that the anti-symmetric projector state $\pi^-$ corresponds to the choice of a uniform distribution $x=y=\frac{1}{3}$.  In this case, the formula yields $\Lambda^2(\pi^-)=\frac{1}{6}$, which matches the value of $\Lambda^2(\omega_{0,1})$ in \eqref{Eq:GME-OO}.

Our numerical results for $\Lambda^2(\tau^{AB}_{x,y}\otimes\tau^{A'B'}_{x, y})$ in $3\otimes 3$ are presented in \Cref{Fig:Singlet_Graph}.  Our results show a triangular region of non-multiplicativity around the anti-symmetric projector state.  Interestingly, for the highly biased subfamily of states
\[p\frac{1}{2}(\op{\Psi^-_{12}}{\Psi^-_{12}}+\op{\Psi^-_{13}}{\Psi^-_{13}})+(1-p)\op{\Psi^-_{23}}{\Psi^-_{23}}\]
we observe non-multiplicativity for values of $(1-p)$ as low as $0.05$.
\begin{figure}[h!]
    \centering
    \includegraphics[width=0.65\linewidth]{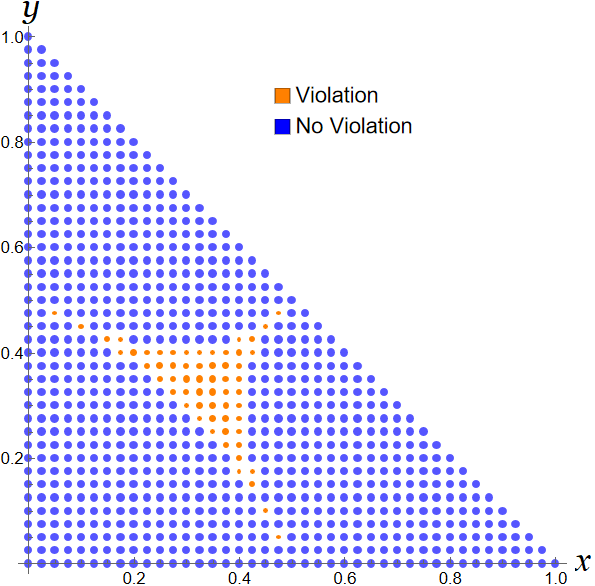}
     \caption{Numerical calculation of $\Lambda^2(\tau_{x,y}^{AB}\otimes\tau_{x,y}^{A'B'})$.}
    \label{Fig:Singlet_Graph}
\end{figure}

\section{GME with real product states}

\label{Sect:real-GME}

We also consider a restricted notion of GME based on how close a given state is to a real product vector.  The definition of the GME is formulated as
\begin{align}
 \Lambda_{\text{real}}^2(\rho^{A_1\cdots A_N})=\max_{\ket{a_1,\cdots,a_N}\in\mc{R}}\bra{a_1,\cdots,a_N}\rho\ket{a_1,\cdots,a_N},\notag
\end{align}
 where $\mc{R}$ is the collection of $N$-partite product states whose coefficients are real in the computational basis.  Of course, it is not surprising that $\Lambda_{\text{real}}^2(\rho^{A_1\cdots A_N})$ and $\Lambda^2(\rho^{A_1\cdots A_N})$ should differ when $\rho^{A_1\cdots A_N}$ itself is not real-valued.  For example, when $\rho^{AB}=\op{0}{0}^A\otimes\op{\wt{+}}{\wt{+}}^{B}$ with $\ket{\wt{+}}=\frac{1}{\sqrt{2}}(\ket{0}+i\ket{1})$, then $\Lambda^2(\op{0}{0}\otimes\op{\wt{+}}{\wt{+}})=1$ while $\Lambda^2_{\text{real}}(\op{0}{0}\otimes\op{\wt{+}}{\wt{+}})=\frac{1}{2}$. What is perhaps surprising, however, is that $\Lambda_{\text{real}}^2(\rho^{A_1\cdots A_N})$ and $\Lambda^2(\rho^{A_1\cdots A_N})$ can differ when $\rho^{A_1\cdots A_N}$ itself is real-valued.  Note that the full class of $(O\otimes O)$-invariant states $\omega_{x,y}$ and the two-level antisymmetric states $\tau_{x,y}$ are both real in the computational basis, and it is interesting to consider whether this phenomenon holds for these states.

 One may claim that the multiplicative property breaks down for the real GME with states that are complex in the computational basis, but we can show this is not always the case with a simple example. For instance, we can define the separable state
 \begin{align}
     \rho^{AB} = \dfrac{1}{4} \left\{ \mathbb{I} \otimes \mathbb{I} - (3/4) Y \otimes Y - (1/4) Z \otimes Z \right\},
 \end{align}
 which is real in the computational basis. We can calculate its local GME to be $5/16$, obtained by the state $\ket{00}$, and show that two copies of $\rho^{AB}$ can achieve a global optimum of $13/128$ for the state $\ket{\Psi^-} \otimes \ket{\Psi^-}$. Notice that $(13/128) - (5/16)^2 = (1/256) > 0$ would imply that the multiplicative property is violated in this instance. This would mean that real two-qubit density operators can still have a non-multiplicative GME even when the optimization is restricted over rebits. 
 
There are a few motivations for considering the real GME.  First, there is an active area of research that studies real theories of quantum mechanics.  In such models, different measures of entanglement emerge based on the notion of real bits or dits (``rebits'' or ``redits'') \cite{Wootters_2014, Stueckelberg_1960, Moretti_2017, Pashaev_2023, Alde_2023}.  This notion of real GME can be made operational from a resource-theoretic perspective in which real-valued objects are considered free \cite{Hickey-2018a, Wu-2021a}, and measures like $\Lambda_{\text{real}}^2(\rho^{A_1\cdots A_N})$ capture how far a given state is from the set of free states.  A second and more practical motivation is that $\Lambda_{\text{real}}^2(\rho^{A_1\cdots A_N})$ is a simpler quantity to compute since its parameter space is smaller. We have especially found faster time-to-solution when maximizing $\Lambda_{\text{real}}^2(\rho^{A_1\cdots A_N})$ in our numerical optimization as compared to $\Lambda^2(\rho^{A_1\cdots A_N})$. This can be quite useful for finding new instances of GME non-multiplicativity since $\Lambda_{\text{real}}^2(\rho^{AB}\otimes\rho^{A'B'})$ provides a lower bound on $\Lambda^2(\rho^{AB}\otimes\rho^{A'B'})$.

For $(O\otimes O)$-invariant states, a similar analysis to \Cref{Sect:OO-states} can be performed in the calculation of $\Lambda_{\text{real}}^2(\omega_{x,y})$. A key difference now is that $\ket{\alpha^*}=\ket{\alpha}$ when restricted to real numbers. This allows for fewer attainable maxima, and we have
\begin{align}
    \label{Eq:GME-OO-real}
\Lambda^2_{\text{real}}(\omega_{x,y})=\max\begin{cases}\frac{2y +d(x+y)}{d(d-1)(d+2)},\\ \frac{1}{3}-\frac{y}{3}-\frac{x(d-2)(d+3)}{3(d-1)(d+2)}\end{cases}.
\end{align}
The two-copy real GME is depicted in \Cref{Fig:OO_Real_Optimization} and shows some notable differences compared to the full GME.  Most notably, we see that non-multiplicativity of $\Lambda^2_{\text{real}}(\omega_{x,y})$ arises even for some separable states.  We also see that the overall number of instances of non-multiplicativity is much larger compared to \Cref{Fig:OO_Complex_Optimization}.

For both families of states, the complex non-multiplicativity region of the GME can be found when restricting the global optimization to real vectors and keeping the local optimizations over the complex space. This is much faster to calculate since we reduce the parameter space in half for each complex vector in the inner product when calculating the global GME, which reduces the number of terms in the objective function by a factor of one-fourth. So instead of finding regions of violation of the GME over the entire complex space, we should restrict to optimizations over the reals for the global objective functions while allowing the local objective functions to be optimized over the complex space. A violation here will surely lead to a violation when the global objectives are not restricted to real vectors since they are a subset of complex vectors. Another thing to note is how simple the local optima were to determine analytically. This makes finding multiplicativity violations much faster and even feasible in cases where considering complex vectors prevent numerical methods to complee their global optimization.
\begin{figure}[h!]
        \centering
        \includegraphics[width=.35\textwidth]{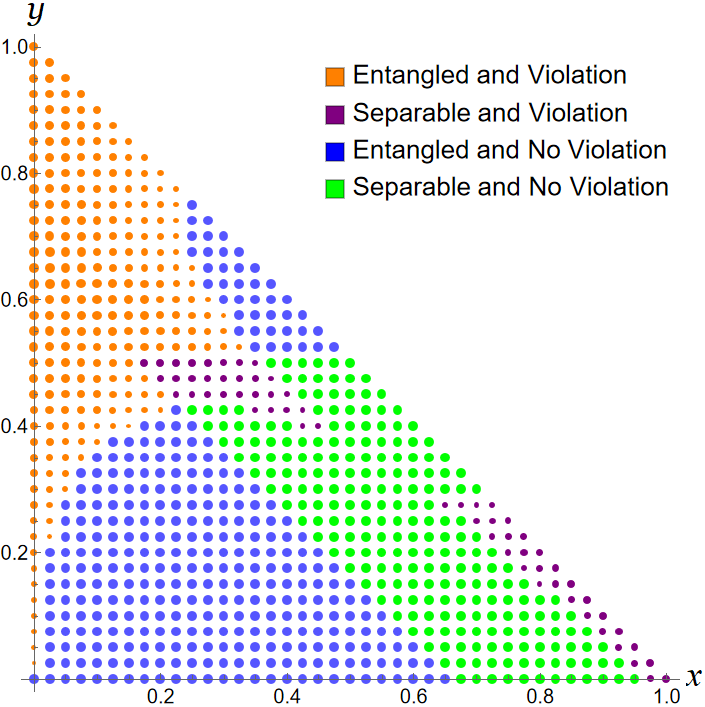}
        \caption{Numerical calculation of $\Lambda_{\text{real}}^2(\omega_{x,y}^{AB}\otimes\omega_{x,y}^{A'B'})$. }
        \label{Fig:OO_Real_Optimization}
\end{figure}

In contrast to the states $\omega_{x,y}$, we did not find any difference in the real GME when considering the states $\tau_{x,y}$.  Since the maximum in \eqref{Eq:GME-convex-asym} can be obtained with real computational basis states, we have that
\begin{align}
    \Lambda^2_{\text{real}}=\frac{1}{2}\max\{x,y,1-x-y\}.
\end{align}
Likewise, when computing $\Lambda^2_{\text{real}}(\tau_{x,y}^{AB}\otimes\tau_{x,y}^{A'B'})$ we found that the results completely matched those shown in \Cref{Fig:Singlet_Graph}.

\section{Conclusion}


In this work, we have explored the multiplicativity of the GME by focusing on specific symmetric families of two-qutrit states.  We found instances of non-multiplicativity in states with large support on the anti-symmetric subspace.  To achieve rigorous verification of the computed bounds, we employed the BARON global optimization solver, which provided a certificate of optimality for both local and global objective functions. A key advantage of our approach was the ability to bypass convex roof extensions by utilizing purifications of mixed states, allowing us to directly relate their GME via trace-based overlap calculations. This ensures that violations of the multiplicative property were both detectable and quantitatively significant, based on predefined thresholds.

Going forward, there are many open problems for further investigation.  First, it would be satisfying to fully resolve the question of GME multiplicity in two-qubit states.  Using the BARON optimization, we explored this question and found no examples of non-multiplicativity.  Another interesting direction would be to study the GME in network entangled states \cite{Navascues-2020a}, which are formed by distributing two or more multipartite entangled states to different subsets of parties.  The natural starting point here would be the triangle network \cite{Kraft-2021a}, and we have not been able to find any examples of GME non-multiplicativity for states formed in this network.  Finally, our study of real GME demonstrates again a sharp qualitative difference between real and complex theories of quantum entanglement.  It is interesting to consider what other differences may emerge between these theories when considering questions of multiplicativity/additivity, and whether this has any operational significance in quantum information theory.

\section{Acknowledgments}
This material is based upon work supported by the U.S. Department of Energy, Office of Science, Accelerated Research in Quantum Computing, Fundamental Algorithmic Research toward Quantum Utility (FAR-Qu), and the National Quantum Information Science Research Centers.

\appendix
\section{Multiplicativity of separable states}\label{app:more}

Here we prove that the GME is strongly multiplicative for all separable states.  An $N$-partite state $\rho^{A_1\cdots A_N}$ is fully separable if it can be expressed as a convex combination of product states,
\begin{align}
\label{Eq:fully-separable}
    \rho^{A_1\cdots A_N}=\sum_{i}p_i\op{a_{1,i},\cdots,a_{N,i}}{a_{1,i},\cdots,a_{N,i}}.
\end{align}
We have that $\Lambda^2(\rho^{A_1\cdots A_N})<1$ whenever there is more than one term in this sum, i.e. whenever $\rho^{A_1\cdots A_N}$ is genuinely mixed.  Here we prove that the GME is strongly multiplicative for fully separable states.
\begin{proposition}
\label{Prop:Separable-GME-multiplicative}
    If $\rho^{A_1\cdots A_N}$ is fully separable and $\sigma^{A_1'\cdots A_N'}$ is any other state, then
    \begin{align}
        \Lambda^2(\rho\otimes \sigma)=\Lambda^2(\rho)\Lambda^2(\sigma).
    \end{align}
\end{proposition}

\begin{proof}
    Consider local states $\ket{\psi_k}^{A_kA_k'}$ for $k=1,\cdots, N$ such that 
    \[\Lambda^2(\rho\otimes\sigma)=\bra{\psi_1,\psi_2,\cdots,\psi_N}\rho\otimes\sigma\ket{\psi_1,\psi_2,\cdots,\psi_N}.\]
    Since $\rho^{A_1\cdots A_N}$ is separable, it can be expressed as $\rho=\sum_i p_i\bigotimes_{k=1}^N\op{a_{k,i}}{a_{k,i}}^{A_k}$.  Define the states
    \[\ket{\alpha_{k,i}}^{A_k'}:={}^{A_k}\!\ip{a_{k,i}}{\psi_k}^{A_{k}A_{k}'}/\sqrt{\gamma_{k,i}}\]
    where $\gamma_{k,i}=\bra{a_{k,i}}\tr_{A_k'}(\op{\psi_k}{\psi_k})\ket{a_{k,i}}$.  Then
    \begin{align}
    \label{Eq:separable-ineq}
&\bra{\psi_1,\psi_2,\cdots,\psi_N}\rho\otimes\sigma\ket{\psi_1,\psi_2,\cdots,\psi_N}   \notag\\
&\qquad=\sum_i p_i\prod_{k=1}^n\gamma_{k,i}\bra{\alpha_{1,i},\cdots,\alpha_{N,i}}\sigma\ket{\alpha_{1,i},\cdots,\alpha_{N,i}}\notag\\
&\qquad\leq \Lambda^2(\sigma)\sum_i p_i\prod_{k=1}^n\gamma_{k,i}.
    \end{align}
The final observation is that
\begin{align}
    &\sum_i p_i\prod_{k=1}^n\gamma_{k,i}\notag\\
    &\;\;=\tr\left(\bigotimes_{k=1}^N\tr_{A_k'}(\op{\psi_k}{\psi_k})\sum_{i}p_i\bigotimes_{k=1}^N\op{a_{k,i}}{a_{k,i}}\right)\notag\\
    &\;\;\leq\Lambda^2(\rho).
\end{align}
Putting together the previous two inequalities yields $\Lambda^2(\rho\otimes \sigma)=\Lambda^2(\rho)\Lambda^2(\sigma)$.
\end{proof}
As a consequence of this proposition, we see that an $N$-partite pure state $\ket{\Psi}^{A_1\cdots A_N}$ is strongly multiplicative if any of its $(N-1)$-partite reduced states are separable.  For example, the class of $N$-qudit GHZ states have a strongly multiplicative GME \cite{Cerf-2002a, Ryu-2013a}.

\vfill

\small

\noindent\framebox{\parbox{0.97\linewidth}{
The submitted manuscript has been created by UChicago Argonne, LLC, Operator of
Argonne National Laboratory (``Argonne''). Argonne, a U.S.\ Department of
Energy Office of Science laboratory, is operated under Contract No.\
DE-AC02-06CH11357.
The U.S.\ Government retains for itself, and others acting on its behalf, a
paid-up nonexclusive, irrevocable worldwide license in said article to
reproduce, prepare derivative works, distribute copies to the public, and
perform publicly and display publicly, by or on behalf of the Government.  The
Department of Energy will provide public access to these results of federally
sponsored research in accordance with the DOE Public Access Plan.
http://energy.gov/downloads/doe-public-access-plan.}}

\newpage

\bibliography{GME_references}

\end{document}